\documentclass[a4paper,12pt,eqno]{article}
\usepackage{amsthm}
\usepackage{latexsym,amssymb,amsfonts,amsmath}
\usepackage{mathrsfs}
\usepackage{graphicx}
\usepackage{color}
\numberwithin{equation}{section}

\newtheorem{thm}{Theorem}[section]
\newtheorem{lem}[thm]{Lemma}
\newtheorem{cor}[thm]{Corollary}
\newtheorem{pro}[thm]{Proposition}
\newtheorem{rem}[thm]{Remark}

\newtheorem{exm}[thm]{Example}

\newcommand{\RM}{\mathbb{R}}

\newcommand{\CM}{\mathbb{C}}

\newcommand{\HM}{\mathbb{H}}
\newcommand{\Mat}{\operatorname{Mat}}
\newcommand{\RP}{\operatorname{Re}}

\newcommand{\Spec}{\operatorname{Spec}}

\title{{\Large {\bf Quaternionic quantum walks of Szegedy type 
and zeta functions of graphs}}
\author{
{\small Norio Konno}\\
{\scriptsize Department of Applied Mathematics, 
Faculty of Engineering, 
Yokohama National University}\\
{\scriptsize Hodogaya, Yokohama 240-8501, Japan, Email: konno@ynu.ac.jp}\\
{\small Kaname Matsue}\\
{\scriptsize Institute of Mathematics for Industry, Kyushu University}\\
{\scriptsize International Institute for Carbon-Neutral Energy Research (WPI-${\rm I^2}$CNER), Kyushu University}\\
{\scriptsize Motooka, Nishi-ku, Fukuoka 819-0395, Japan, Email: kmatsue@imi.kyushu-u.ac.jp}\\
{\small Hideo Mitsuhashi}\\
{\scriptsize Department of Applied Informatics, Faculty of Science and Engineering, Hosei University}\\
{\scriptsize  Koganei, Tokyo 184-8584, Japan, Email: hmitsu@hosei.ac.jp}\\
{\small Iwao Sato}\\
{\scriptsize Oyama National College of Technology}\\
{\scriptsize Oyama, Tochigi 323-0806, Japan, Email: isato@oyama-ct.ac.jp}\\}
}
\date{\empty }
\begin{document}
\maketitle

\par\noindent
\begin{small}
\par\noindent
{\bf Abstract}. 
We define a quaternionic extension of the Szegedy walk on a graph and study 
its right spectral properties. 
The condition for the transition matrix of the quaternionic Szegedy walk on a graph 
to be quaternionic unitary is given. 
In order to derive the spectral mapping theorem for the quaternionic Szegedy walk, 
we derive a quaternionic extension of the determinant expression 
of the second weighted zeta function of a graph. 
Our main results determine explicitly all the right eigenvalues of the 
quaternionic Szegedy walk by using complex right eigenvalues of the corresponding 
doubly weighted matrix. 
We also show the way to obtain eigenvectors corresponding to 
right eigenvalues derived from those of doubly weighted matrix. 
  
\footnote[0]{
{\it Abbr. title:} Quaternionic quantum walks of Szegedy type }
\end{small}

\setcounter{equation}{0}
\section{Introduction}\label{sec:Intro}
The quantum walk has been studied over the last two decades in various fields such as 
quantum information theory and mathematics. Quantum walks are classified into two types: discrete-time quantum walks and continuous-time 
quantum walks. We treat the discrete-time quantum walk on a graph in this study. 
The discrete-time quantum walk in one dimension lattice was intensively studied by Ambainis et al. \cite{AmbainisEtAl2001}. One of its most striking properties is the spreading property of the walker. The standard deviation of the position grows linearly in time, quadratically faster than classical random walk. On the other hand, a walker stays at the starting position, i.e., localization occurs. 
The reviews and books on quantum walks are Kempe \cite{Kempe2003}, Konno \cite{Konno2008b}, Venegas-Andraca \cite{Venegas2013}, Cantero et al. \cite{CGMV}, Manouchehri and Wang \cite{MW2013}, Portugal \cite{Port2013} for examples. 

One of the most important discrete-time quantum walks on graphs is the Grover walk on a graph. 
The Grover walk on a graph was first formulated in \cite{Grover1996} for complete graphs. 
After Grover's approach, several researchers, for example \cite{Watrous2001,Ambainis2003b,Szegedy2004} proposed quantum walks on graphs. 
Emms et al. \cite{EmmsETAL2006} and Godsil and Guo \cite{GG2010} applied the spectral property of the Grover walk to graph isomorphism problem. 
In the present article, we focus on the discrete-time quantum walk on a graph based on 
the quantum search algorithm proposed by Szegedy \cite{Szegedy2004}. 
This walk arises naturally from a Markov chain, and is called the Szegedy walk. 

Zeta functions of graphs originate from the Ihara zeta function for a regular graph 
introduced by Ihara. In \cite{Ihara1966}, Ihara defined a $\mathfrak{p}$-adic analogue of 
the Selberg zeta function associated to a certain kind of discrete subgroups of 
the $2$ by $2$ projective linear group over $\mathfrak{p}$-adic fields. 
After that, the Ihara zeta function has been extensively studied as a zeta function of a graph by many researchers 
\cite{Ihara1966,Serre,Sunada1986,Sunada1988,Hashimoto1989,Bass1992,ST1996,FZ1999,KS2000}. 
Sato \cite{Sato} defined a new zeta function (the second weighted zeta function) of a graph 
by modifying a determinant expression of the (first) weighted zeta function which had been 
proposed by Mizuno and Sato \cite{MS2004}. This new zeta function and its 
determinant formula play essential roles in the determination of eigenvalues of the quantum 
walk on a graph in Konno and Sato \cite{KS2012} and in another proof of the Smilansky's formula 
\cite{Smilansky2007} for the characteristic polynomial of the bond scattering matrix of a graph in 
Mizuno and Sato \cite{MS2008}. 
Ren et al.  \cite{RenETAL} found an interesting relation between 
the Ihara zeta function and the Grover walk on a graph, and
showed that the positive support of the transition matrix of the Grover walk is equal to 
the Perron-Frobenius operator (the edge matrix) related to the Ihara zeta function. 
Konno and Sato \cite{KS2012} gave the characteristic polynomials of the transition matrix 
of the Grover walk and its positive support by using the second weighted zeta function 
and the Ihara zeta function, 
and thus obtained an alternative proof of the results on spectra by Emms et al. \cite{EmmsETAL2006}. 
Higuchi et al. \cite{HKSS2013,HKSS2014} introduced the twisted Szegedy walk on a graph 
and obtained the spectral mapping theorem and eigenspaces for the twisted Szegedy walk. 

The quaternion was discovered by Hamilton in 1843. 
Since quaternions do not commute with each other in general, 
it is necessary to treat left eigenvalues and right eigenvalues separately. 
The eigenvalue problem for quaternionic matrices has been investigated 
for nearly a century by a number of researchers. 
Right eigenvalues are well studied, 
in contrast, left eigenvalues are less known and not easy to handle as commented in 
\cite{Zhang1997}. 
In this paper, we concentrate only on the right spectrum of quaternionic quantum walk. 

Recently, an extension of quantum walk to the case of quaternions was established 
by Konno \cite{Konno2015} and its various properties were studied. 
Along this line, Konno et al. \cite{KMS2016} gave a definition of 
a discrete-time quaternionic quantum walk on a graph which can be viewed as 
a quaternionic extension of the Grover walk, 
and obtained the spectral mapping theorem for the walk. 
One of the important backgrounds of quaternionic quantum walk is quaternionic quantum mechanics. 
The possibility of quaternionic quantum mechanics was first pointed out by 
Birkhoff and von Neumann in \cite{BN1936}. 
After that, the subject was studied further by Finkelstein, Jauch, and Speiser \cite{FJS1959}, 
and more recently by Adler \cite{Adler1995} and others. 
For a detailed account including the significance of quaternionic quantum mechanics, 
we recommend Adler's book \cite{Adler1995}. 
As well as quaternionic quantum mechanics, 
we can expect many outcomes through the study on quaternionic quantum walk. 

In the present paper, we define a quaternionic generalization of the Szegedy walk 
(the quaternionic Szegedy walk) on a graph. 
We derive the unitary condition of the transition matrix of the quaternionic Szegedy walk. 
Furthermore, we investigate the right eigenvalue problem, 
and derive all the right eigenvalues of the 
quaternionic transition matrix from those of the doubly weighted matrix 
which can be obtained easily from quaternionic weights on the graph. 
We also show the way to obtain right eigenvectors corresponding to the right eigenvalues 
derived from those of the doubly weighted matrix. 

The rest of the paper is organized as follows. 
In Section \ref{sec:RightEigenOfQuatMat}, we explain some properties of quaternionic matrices 
and right eigenvalues of quaternionic matrices. 
Especially, we give the way to obtain all the right eigenvalues of a quaternionic matrix. 
In Section \ref{sec:IharaZeta}, we provide a summary of the Ihara zeta function and the second weighted zeta function of a graph, and present their determinant expressions. 
In Section \ref{sec:QuatDetFormula}, we derive a quaternionic extension of the determinant formula 
for the second weighted zeta function. 
The result in this section plays a key role in determination of right spectra 
in Section \ref{sec:TransMat}. 
In Section \ref{sec:TransMat}, we give a brief account of the discrete-time quantum walk on a graph 
and define a generalization of it to the case of quaternions, which we call 
the quaternionic Szegedy walk on a graph. 
We also explain the general framework of quaternionic 
quantum mechanics to the minimum extent necessary for the definition of 
quaternionic quantum walks. 
We derive the unitary condition on the transition matrix and determine all the right eigenvalues 
of a quaternionic quantum walk by using those of the corresponding doubly weighted matrix. 
In Section \ref{sec:Eigenvectors}, we derive the way to obtain eigenvectors corresponding to 
right eigenvalues of the transition matrix derived from those of doubly weighted matrix. 
In addition, we give an example of quaternionic Szegedy walk and their right eigenvalues and 
eigenvectors. 

\section{Right eigenvalues and root subspaces of a quaternionic matrix}\label{sec:RightEigenOfQuatMat}
In this section, we give a brief account of quaternions and quaternionic matrix. 
Particularly, we mainly explain eigenvalues and root subspaces of a quaternionic matrix. 
Since quaternions do not mutually commute in general, 
we must treat left eigenvalues and right eigenvalues separately. 
As in \cite{Adler1995}, the states of quaternionic quantum mechanics is described by vectors of a right Hilbert space over the quaternion field. 
Therefore, we deal only with right eigenvalues in this paper. 
We shall review the way to obtain all the right eigenvalues of a quaternionic matrix. 
Expositions of these contents can be found in \cite{Aslaksen1996,KMS2016,DeLeoEtAl,Rodman2014}. 
\cite{Zhang1997} gives an overview of the quaternionic matrix theory. 
The detailed account of recent development of quaternionic linear algebra can be found in 
\cite{Rodman2014}.  

Let $\HM$ be the set of quaternions. $\HM$ is a noncommutative associative 
algebra over $\RM$, whose underlying real vector space has dimension $4$ 
with a basis $1,i,j,k$ which satisfy the following relations: 
\[
i^2=j^2=k^2=-1,\quad ij=-ji=k,\quad jk=-kj=i,\quad ki=-ik=j.
\]
For $x=x_0+x_1i+x_2j+x_3k{\;\in\;}\HM$, $x^*$ denotes the {\it conjugate} of $x$ in $\HM$ 
which is defined by $x^*=x_0-x_1i-x_2j-x_3k$.
We call $|x|=\sqrt{xx^*}=\sqrt{x^*x}=\sqrt{x_0^2+x_1^2+x_2^2+x_3^2}$ the {\it norm} of $x$. 
$\HM$ constitutes a skew field since $x{\;\in\;}\HM$, $x^{-1}=x^*/|x|^2$ for every nonzero element $x$. 
We denote by $\HM^*$ the multiplicative group of nonzero quaternions. 

$\Mat(m{\times}n,\HM)$ denotes the set of $m{\times}n$ quaternionic matrices and 
$\Mat(n,\HM)$ the set of $n{\times}n$ quaternionic square matrices. A quaternion $x$ can be identified with an element of $\Mat(1,\HM)$. 
For every quaternionic matrix ${\bf M}{\;\in\;}\Mat(m{\times}n,\HM)$, 
we can write ${\bf M}={\bf A}+j{\bf B}$ 
uniquely where ${\bf A},{\bf B}{\;\in\;}\Mat(m{\times}n,\CM)$. 
${\bf A}$ and ${\bf B}$ are called the {\it simplex part} and the {\it perplex part} 
of ${\bf M}$ respectively. 
Then, the $\RM$-linear map $\psi$ from $\Mat(m{\times}n,\mathbb{H})$ to $\Mat(2m{\times}2n,\mathbb{C})$ is defined as follows: 
\[
\psi : \Mat(m{\times}n,\mathbb{H}){\;\longrightarrow\;}\Mat(2m{\times}2n,\mathbb{C})
\quad{\bf M}{\;\mapsto\;}\begin{pmatrix}{\bf A}&-\overline{\bf B}\\{\bf B}&
\overline{\bf A}\end{pmatrix},
\]
where $\overline{\bf A}$ is the complex conjugate of ${\bf A}$. 
$\psi$ is multiplicative in the following sense: 

\begin{lem}[\cite{KMS2016}, Lemma 2.1]\label{PsiLem}
Let ${\bf M}{\;\in\;}\Mat(m{\times}n,\HM)$ and ${\bf N}{\;\in\;}
\Mat(n{\times}m,\HM)$. Then 
\[ \psi({\bf M}{\bf N})=\psi({\bf M})\psi({\bf N}). \]
\end{lem}

In the case of square matrices, we have 

\begin{pro}[\cite{KMS2016}, Proposition 2.2]\label{ProInjectivityOfPsi}
If $m=n$, then $\psi$ is an injective $\mathbb{R}$-algebra homomorphism. 
\end{pro}

For a quaternionic square matrix ${\bf M}=({\bf M}_{rs}){\;\in\;}\Mat(n,\HM)$, 
the {\it conjugate} ${\bf M}^*=(({\bf M}^*)_{rs})$ is defined by 
$({\bf M}^*)_{rs}=({\bf M}_{sr})^*$. $\psi({\bf M}^*)=\psi({\bf M})^*$ holds, where the right-hand side denotes the conjugate 
transpose of the complex matrix. 
A quaternionic square matrix ${\bf M}$ is said to be {\it unitary} if 
${\bf M}^*{\bf M}={\bf M}{\bf M}^*={\bf I}$. 
By Proposition \ref{ProInjectivityOfPsi}, ${\bf M}^*{\bf M}={\bf I}$ implies ${\bf M}{\bf M}^*={\bf I}$. 

We consider $\HM^n$ as a right $\HM$-vector space. 
${\bf v}_1,\hdots,{\bf v}_m{\;\in\;}\HM^n$ are said to be {\it $\HM$-linearly independent} 
if ${\bf v}_1a_1+{\cdots}+{\bf v}_ma_m={\bf 0}$ for $a_1,{\hdots},a_m{\;\in\;}\HM$ 
implies $a_1={\cdots}=a_m=0$. 
We shall show a useful criterion for quaternionic vectors to be $\HM$-linearly independent, which is a slightly different version of \cite{Rodman2014}, Proposition 3.4.2. 

\begin{pro}[\cite{Rodman2014}, Proposition 3.4.2]\label{ProHLinearIndep}
${\bf v}_1,\hdots,{\bf v}_m{\;\in\;}\HM^n$ are $\HM$-linearly independent if and only if 
the columns of $\psi(\begin{pmatrix} {\bf v}_1 & \cdots & {\bf v}_m \end{pmatrix})$ 
are linearly independent over $\CM$.   
\end{pro}

\begin{proof}
Suppose $\psi(\begin{pmatrix} {\bf v}_1 & \cdots & {\bf v}_m \end{pmatrix})$ 
are linearly independent over $\CM$. 
If ${\bf v}_1a_1+{\cdots}+{\bf v}_ma_m={\bf 0}$, then 
\[
\begin{pmatrix} {\bf v}_1 & \cdots & {\bf v}_m \end{pmatrix}
\begin{pmatrix} a_1 \\ \vdots \\ a_m \end{pmatrix}
=\begin{pmatrix} 0 \\ \vdots \\ 0 \end{pmatrix}.
\]
Applying $\psi$ to both sides, we have 
\[
\psi(\begin{pmatrix} {\bf v}_1 & \cdots & {\bf v}_m \end{pmatrix})
\begin{pmatrix} a_1^S & -\overline{a_1^P} \\ \vdots & \vdots \\ a_m^S & -\overline{a_m^P} \\ 
a_1^P & \overline{a_1^S} \\ \vdots & \vdots \\ a_m^P & \overline{a_m^S}\end{pmatrix}
=\begin{pmatrix} 0 & 0 \\ \vdots & \vdots \\ 0 & 0 \\ 0 & 0 \\ \vdots & \vdots \\ 0 & 0 \end{pmatrix},
\]
where $a_r=a_r^S+ja_r^P$ is the symplectic decomposition of $a_r$. 
Since columns of the matrix $\psi(\begin{pmatrix} {\bf v}_1 & \cdots & {\bf v}_m \end{pmatrix})$ are 
linearly independent, $a_1^S=a_1^P={\cdots}=a_m^S=a_m^P=0$ and thus we obtain $a_1={\cdots}=a_m=0$. 

Conversely, suppose ${\bf v}_1,\hdots,{\bf v}_m{\;\in\;}\HM^n$ are $\HM$-linearly independent. 
If 
\[
\psi(\begin{pmatrix} {\bf v}_1 & \cdots & {\bf v}_m \end{pmatrix})
\begin{pmatrix} b_1 \\ \vdots \\ b_m \\ 
c_1 \\ \vdots \\ c_m \end{pmatrix}
=\begin{pmatrix} 0 \\ \vdots \\ 0 \\ 0 \\ \vdots \\ 0 \end{pmatrix},
\]
then, letting ${\bf v}_r={\bf v}_r^S+j{\bf v}_r^P$ be the symplectic decomposition, 
we have 
\begin{equation*}
\begin{split}
&{\bf v}_1^Sb_1+{\cdots}+{\bf v}_m^Sb_m
-\overline{{\bf v}_1^P}c_1-{\cdots}-\overline{{\bf v}_m^P}c_m={\bf 0},\\
&{\bf v}_1^Pb_1+{\cdots}+{\bf v}_m^Pb_m
+\overline{{\bf v}_1^S}c_1+{\cdots}+\overline{{\bf v}_m^S}c_m={\bf 0},
\end{split}
\end{equation*}
whence it follows that 
\begin{equation}\label{EqnLinearlyIndep}
\begin{split}
&{\bf v}_1^Sb_1+{\cdots}+{\bf v}_m^Sb_m
-\overline{{\bf v}_1^P}c_1-{\cdots}-\overline{{\bf v}_m^P}c_m={\bf 0},\\
&\overline{{\bf v}_1^P}jb_1+{\cdots}+\overline{{\bf v}_m^P}jb_m
+{\bf v}_1^Sjc_1+{\cdots}+{\bf v}_m^Sjc_m={\bf 0},
\end{split}
\end{equation}
since $\overline{x}j=jx$ for every $x{\;\in\;}\CM$. 
Summing up both sides of (\ref{EqnLinearlyIndep}), we obtain 
\[
{\bf v}_1(b_1+jc_1)+{\cdots}+{\bf v}_m(b_m+jc_m)={\bf 0}.
\]
Since ${\bf v}_1,{\hdots},{\bf v}_m$ are $\HM$-linearly independent, 
$b_1={\cdots}=b_m=c_1={\cdots}=c_m=0$ holds. 
\end{proof}

For a quaternionic square matrix ${\bf M}=({\bf M}_{rs}){\;\in\;}\Mat(n,\HM)$, 
a quaternion $\lambda$ and a nonzero vector ${\bf v}{\;\in\;}\HM^n$ are said to be 
a {\it right eigenvalue} of ${\bf M}$ and a {\it right eigenvector} 
corresponding to $\lambda$, respectively if ${\bf Mv}={\bf v}\lambda$. 
The set of all right eigenvalues of ${\bf M}$ is denoted by $\sigma_r({\bf M})$. 
Since ${\bf M}({\bf v}x)={\bf v}{\lambda}x={\bf v}x(x^{-1}{\lambda}x)$ 
for any $x{\;\in\;}\mathbb{H}^*=\mathbb{H}-\{0\}$, the conjugate class 
$\lambda^{\mathbb{H}^*}=\{x^{-1}{\lambda}x\;|\;x{\;\in\;}\mathbb{H}^*\}$ is contained in $\sigma_r({\bf M})$ and ${\bf v}x$ is a right eigenvector corresponding to the eigenvalue 
$x^{-1}{\lambda}x$. 
It is known that 
right eigenvalues of a square quaternionic matrix are given by the following 
manner. The detail of its proof can be found in, for example \cite{KMS2016}.  

\begin{thm}[\cite{KMS2016}, Theorem 2.8]\label{QuatEigen}
For any ${\bf M}{\;\in\;}\Mat(n,\mathbb{H})$, 
there exist $2n$ complex right eigenvalues $\lambda_1,{\ldots},\lambda_n,$ 
$\overline{\lambda_1},{\ldots},\overline{\lambda_n}$ of ${\bf M}$ 
counted with multiplicity,
which can be calculated by solving $\det({\lambda}{\bf I}_{2n}-\psi({\bf M}))=0$. 
The set of right eigenvalues $\sigma_r({\bf M})$ is given by 
$\sigma_r({\bf M})=\lambda_1^{\mathbb{H}^*}{\cup}{\;\cdots\;}{\cup}\lambda_n^{\mathbb{H}^*}$, 
where $\lambda^{\mathbb{H}^*}=\{x^{-1}{\lambda}x\;|\;x{\;\in\;}\mathbb{H}^*\}$ is the 
set of all the conjugations of $\lambda$ by nonzero quaternions. 
\end{thm}

\begin{rem}\label{RemQuatEigen}
{\rm (1)} $\lambda_r^{\mathbb{H}^*}{\cap}\lambda_s^{\mathbb{H}^*}=\emptyset$ if 
$\lambda_r$ equals neither $\lambda_s$ nor $\overline{\lambda_s}$. \\
{\rm (2)} If ${\bf v}$ is an eigenvector corresponding to a complex right eigenvalue 
$\lambda$, then ${\bf v}j$ is an eigenvector corresponding to $\overline{\lambda}$, 
since ${\bf M}{\bf v}j={\bf v}{\lambda}j={\bf v}j\overline{\lambda}$. 
\end{rem}

Let $p^{({\bf M})}(y)$ be the real coefficient polynomial of minimal degree 
which satisfies the following conditions: 
\begin{enumerate}
\renewcommand{\labelenumi}{\rm (M\arabic{enumi})}
\item
$p^{({\bf M})}(y)$ is monic, namely, the leading coefficient equals $1$. 
\item
$p^{({\bf M})}({\bf M})={\bf O}_n$.
\end{enumerate}
Such a polynomial exists since $\Mat(n,\HM)$ is a finite dimensional vector space 
over $\RM$. 
$p^{({\bf M})}(y)$ is called the {\it minimal polynomial} of ${\bf M}$. 
As in the case of complex matrices, the minimal polynomial uniquely exists and 
$p^{({\bf M})}(y)$ divides any real coefficient polynomial $p(y)$ such that 
$p({\bf M})={\bf O}_n$. 
Let $p^{({\bf M})}(y)=p^{({\bf M})}_1(y)^{m_1}{\cdots}p^{({\bf M})}_r(y)^{m_r}$ be 
the factorization of $p^{({\bf M})}(y)$ into mutually different monic irreducible 
real coefficient polynomials $p^{({\bf M})}_1(y),{\cdots},p^{({\bf M})}_r(y)$. 
Since all coefficients are real, the factorization of $p^{({\bf M})}(y)$ yields 
only monic irreducible polynomials of degree $1$ or $2$ as factors; 
Each $p^{({\bf M})}_s({\bf M})$ has either 
one real root or two complex roots which are complex conjugate with each other. 
For each $s=1,2,{\ldots},r$, 
the following subspace $\mathscr{M}_{s}$ of $\HM^n$ is called 
the {\it root subspace} of ${\bf M}$: 
\begin{equation}\label{DefRootSubspace}
\mathscr{M}_{s}=\{{\bf v}{\;\in\;}\HM^n\,|\,p^{({\bf M})}_s({\bf M})^{m_s}{\bf v}={\bf o}\}.
\end{equation}
We notice that $\mathscr{M}_s$ is ${\bf M}$-invariant, namely, 
${\bf M}\mathscr{M}_s{\,\subseteq\,}\mathscr{M}_s$ since 
\[
p^{({\bf M})}_s({\bf M})^{m_s}{\bf M}={\bf M}p^{({\bf M})}_s({\bf M})^{m_s} 
\]
holds. It is known that $\HM^n$ is the direct sum of root spaces as follows: 

\begin{pro}[\cite{Rodman2014}, Proposition 5.1.4]
Let $p^{({\bf M})}(y)$ be the minimal polynomial of the matrix ${\bf M}{\;\in\;}\Mat(n,\HM)$. 
If $p^{({\bf M})}(y)=p^{({\bf M})}_1(y)^{m_1}{\cdots}p^{({\bf M})}_r(y)^{m_r}$ is the 
factorization of $p^{({\bf M})}(y)$ into mutually different monic irreducible 
real coefficient polynomials, then 
\[
\HM^n=\bigoplus_{s=1}^r\mathscr{M}_s,
\]
where $\mathscr{M}_s$ is the root subspace of ${\bf M}$. 
\end{pro}

Since $p^{({\bf M})}({\bf M})={\bf O}_n$ is equivalent to $p^{({\bf M})}(\psi({\bf M}))={\bf O}_{2n}$, 
$p^{({\bf M})}(y)$ is the minimal polynomial of $\psi({\bf M})$. 
Thus, $p^{({\bf M})}(y)$ divides $\det(y{\bf I}_{2n}-\psi({\bf M}))$. 
From Theorem \ref{QuatEigen}, we have
\begin{equation}\label{EqnFactCharacteristicEqn}
\begin{split}
\det(y{\bf I}_{2n}-\psi({\bf M}))&=(y-\lambda_1)(y-\overline{\lambda}_1)
{\cdots}(y-\lambda_n)(y-\overline{\lambda}_n)\\
&=(y-2{\RP}\lambda_1+|\lambda_1|^2){\cdots}(y-2{\RP}\lambda_n+|\lambda_n|^2).
\end{split}
\end{equation}
Each eigenvalue $\lambda$ of $\psi({\bf M})$ is a root of some factor 
$p^{({\bf M})}_s(y)$ of $p^{({\bf M})}(y)$, and the right eigenvector ${\bf v}$ of ${\bf M}$ in $\HM^n$ 
corresponding to $\lambda$ belongs to the root subspace $\mathscr{M}_s$. 
For every $x{\;\in\;}\HM^*$, ${\bf v}x$ also belongs to $\mathscr{M}_s$. 
Thus, for every $\lambda'{\;\in\;}\lambda^{\HM^*}$, all eigenvectors corresponding to $\lambda'$ 
belong to $\mathscr{M}_s$. Especially, eigenvectors corresponding to a complex right eigenvalue 
and its complex conjugate belong to the same root subspace since $-jij=-i$. 
The maximal number of $\HM$-linearly independent right eigenvectors of 
a root subspace $\mathscr{M}$ of a quaternionic square matrix 
${\bf M}{\;\in\;}\Mat(n,\HM)$ is called the {\it geometric multiplicity} 
of $\mathscr{M}$. 
For a right eigenvalue $\lambda{\;\in\;}\sigma_r({\bf M})$, 
the geometric multiplicity of the root subspace corresponding to the 
polynomial of which $\lambda$ is a root is called the geometric multiplicity 
of $\lambda$. We notice the geometric multiplicity of a right eigenvalue is well-defined 
because of Theorem \ref{QuatEigen}, Remark \ref{RemQuatEigen} (1) and the observation 
in the previous paragraph. 
On the other hand, the quaternionic dimension of $\mathscr{M}$ is called the {\it algebraic multiplicity} 
of $\mathscr{M}$. 
For a right eigenvalue $\lambda{\;\in\;}\sigma_r({\bf M})$, 
the algebraic multiplicity of the root subspace corresponding to the 
polynomial of which $\lambda$ is a root is called the algebraic multiplicity 
of $\lambda$. 

We also notice that one can compute a right eigenvector corresponding to a complex eigenvalue 
$\lambda$ of a quaternionic square matrix as follows. 

\begin{pro}[\cite{KMS2016}, Proposition 2.5]\label{EigenCorres}
Let ${\bf M}{\;\in\;}\Mat(n,\HM)$ and $\lambda{\;\in\;}\mathbb{C}$. \\Then 
\begin{equation*}
{\bf M}{\bf v}={\bf v}\lambda{\;\Leftrightarrow\;}\psi({\bf M})
\begin{pmatrix} {\bf u} \\ {\bf w} \end{pmatrix}
=\begin{pmatrix} {\bf u} \\ {\bf w} \end{pmatrix}\lambda
\qquad \text{where\quad} {\bf v}={\bf u}+j{\bf w}\;
({\bf u},{\bf w}{\;\in\;}\mathbb{C}^n).
\end{equation*}
\end{pro}

Here, we shall give some examples of right spectra of quaternionic matrices. For an arbitrary matrix ${\bf M}$, ${}^T\!{\bf M}$ denotes the transpose of ${\bf M}$. 

\begin{exm}
${\bf M}=\begin{pmatrix}1&0\\0&k\end{pmatrix}$\\
$\det({\lambda}{\bf I}_{4}-\psi({\bf M}))=\det\begin{pmatrix}\lambda-1&0&0&0\\0&\lambda&0&i\\
0&0&\lambda-1&0\\0&i&0&\lambda\end{pmatrix}=(\lambda-1)^2(\lambda-i)(\lambda+i)$\\
$=(\lambda-1)^2(\lambda^2+1)$. \\
$p^{({\bf M})}(\lambda)=(\lambda-1)(\lambda^2+1)$. \quad $p^{({\bf M})}_1(\lambda)=\lambda-1$,
$p^{({\bf M})}_2(\lambda)=\lambda^2+1$. \quad 
$\sigma_r({\bf M})=\{1\}{\cup}i^{\mathbb{H}^*}$.\\
$\lambda=1{\;\Rightarrow\;}
\begin{pmatrix}{\bf u}\\{\bf w}\end{pmatrix}={}^T\!\!\begin{pmatrix}1&0&0&0\end{pmatrix},\;
{}^T\!\!\begin{pmatrix}0&0&1&0\end{pmatrix}$. \quad 
${\bf v}={}^T\!\!\begin{pmatrix}1&0\end{pmatrix},\,{}^T\!\!\begin{pmatrix}j&0\end{pmatrix}$. \\
$\lambda=i{\;\Rightarrow\;}
\begin{pmatrix}{\bf u}\\{\bf w}\end{pmatrix}={}^T\!\!\begin{pmatrix}0&1&0&-1\end{pmatrix}$. \quad 
${\bf v}={}^T\!\!\begin{pmatrix}0&1-j\end{pmatrix}$. \\
$\lambda=-i{\;\Rightarrow\;}
\begin{pmatrix}{\bf u}\\{\bf w}\end{pmatrix}={}^T\!\!\begin{pmatrix}0&1&0&1\end{pmatrix}$. \quad 
${\bf v}={}^T\!\!\begin{pmatrix}0&1+j\end{pmatrix}$. \\
$\HM^2=\mathscr{M}_1{\oplus}\mathscr{M}_2$, 
where $\mathscr{M}_1={}^T\!\!\begin{pmatrix}1&0\end{pmatrix}\HM,\;
\mathscr{M}_2={}^T\!\!\begin{pmatrix}0&1-j\end{pmatrix}\HM$. 
\end{exm}

\begin{exm}
${\bf M}=\begin{pmatrix}0&i\\j&0\end{pmatrix}=\begin{pmatrix}0&i\\0&0\end{pmatrix}
+j\begin{pmatrix}0&0\\1&0\end{pmatrix}$\\
$\det({\lambda}{\bf I}_{4}-\psi({\bf M}))=\det\begin{pmatrix}\lambda&-i&0&0\\0&\lambda&1&0\\
0&0&\lambda&i\\-1&0&0&\lambda\end{pmatrix}=(\lambda^2-i)(\lambda^2+i)$\\
$=\left( \lambda-\dfrac{1+i}{\sqrt{2}} \right)
\left(\lambda-\dfrac{-1-i}{\sqrt{2}} \right)\left(\lambda-\dfrac{1-i}{\sqrt{2}} \right)\left(\lambda-\dfrac{-1+i}{\sqrt{2}} \right)
=(\lambda^2-\sqrt{2}\lambda+1)(\lambda^2+\sqrt{2}\lambda+1)$. \\ 
$p^{({\bf M})}(\lambda)=(\lambda^2-\sqrt{2}\lambda+1)(\lambda^2+\sqrt{2}\lambda+1)$. \quad 
$p^{({\bf M})}_1(\lambda)=\lambda^2-\sqrt{2}\lambda+1$, 
$p^{({\bf M})}_2(\lambda)=\lambda^2+\sqrt{2}\lambda+1$.\\  
$\sigma_r({\bf M})=\Big{(}\dfrac{1+i}{\sqrt{2}}\Big{)}^{\mathbb{H}^*}{\cup}
\Big{(}\dfrac{-1-i}{\sqrt{2}}\Big{)}^{\mathbb{H}^*}$.\\
$\lambda=\dfrac{1+i}{\sqrt{2}}{\;\Rightarrow\;}
\begin{pmatrix}{\bf u}\\{\bf w}\end{pmatrix}
={}^T\!\!\begin{pmatrix}1&\dfrac{1-i}{\sqrt{2}}&-1&\dfrac{1-i}{\sqrt{2}}\end{pmatrix}$. 
${\bf v}={}^T\!\!\begin{pmatrix}1-j&\dfrac{1-i+j+k}{\sqrt{2}}\end{pmatrix}$. \\
$\lambda=\dfrac{1-i}{\sqrt{2}}{\;\Rightarrow\;}
\begin{pmatrix}{\bf u}\\{\bf w}\end{pmatrix}
={}^T\!\!\begin{pmatrix}1&-\dfrac{1+i}{\sqrt{2}}&1&\dfrac{1+i}{\sqrt{2}}\end{pmatrix}$. 
${\bf v}={}^T\!\!\begin{pmatrix}1+j&\dfrac{-1-i+j-k}{\sqrt{2}}\end{pmatrix}$. \\
$\lambda=\dfrac{-1-i}{\sqrt{2}}{\;\Rightarrow\;}
\begin{pmatrix}{\bf u}\\{\bf w}\end{pmatrix}
={}^T\!\!\begin{pmatrix}1&-\dfrac{1-i}{\sqrt{2}}&-1&-\dfrac{1-i}{\sqrt{2}}\end{pmatrix}$. 
${\bf v}={}^T\!\!\begin{pmatrix}1-j&\dfrac{-1+i-j-k}{\sqrt{2}}\end{pmatrix}$. \\
$\lambda=\dfrac{-1+i}{\sqrt{2}}{\;\Rightarrow\;}
\begin{pmatrix}{\bf u}\\{\bf w}\end{pmatrix}
={}^T\!\!\begin{pmatrix}1&\dfrac{1+i}{\sqrt{2}}&1&-\dfrac{1+i}{\sqrt{2}}\end{pmatrix}$. 
${\bf v}={}^T\!\!\begin{pmatrix}1+j&\dfrac{1+i-j+k}{\sqrt{2}}\end{pmatrix}$. \\
$\HM^2=\mathscr{M}_1{\oplus}\mathscr{M}_2$, 
where $\mathscr{M}_1={}^T\!\!\begin{pmatrix}1-j&\dfrac{1-i+j+k}{\sqrt{2}}\end{pmatrix}\HM,$ and\\
$\mathscr{M}_2={}^T\!\!\begin{pmatrix}1-j&\dfrac{-1+i-j-k}{\sqrt{2}}\end{pmatrix}\HM$. 
\end{exm}

\section{The second weighted zeta function of a graph}\label{sec:IharaZeta}
We shall review the Ihara zeta function and the second weighted zeta function of a graph in this section. 
In this section, let $G=(V(G)$, $E(G))$ be a finite connected simple graph with the vertex set $V(G)$ and 
the edge set $E(G)$ consisting of unoriented edges $uv$ joining two vertices $u$ and $v$. 
% In this section, we assume that $G$ is simple, that is, $G$ has neither loops nor multiple edges. 
For $uv \in E(G)$, the ordered pair $(u,v)$ denotes the oriented edge from $u$ to $v$ which we call an {\em arc}. 
Let $D(G)=\{\,(u,v),\,(v,u)\,\mid\,uv{\;\in\;}E(G)\}$ and $|V(G)|=n,\;|E(G)|=m,\;|D(G)|=2m$. 
For $e=(u,v){\;\in\;}D(G)$, $o(e)=u$ denotes the {\it origin} and $t(e)=v$ the {\it terminus} 
of $e$ respectively. 
$e^{-1}=(v,u)$ is called the {\em inverse} of $e=(u,v)$. 
The {\em degree} $\deg v$ of a vertex $v$ of $G$ is the number of edges 
incident to $v$. 
A {\em path $P$ of length $\ell$} in $G$ is a sequence 
$P=(e_1, \cdots ,e_{\ell})$ of $\ell$ arcs $e_r \in D(G)\;(r=1,\ldots,\ell)$ with 
$t(e_r)=o(e_{r+1})$ for $r=1,\ldots,\ell-1$. 
$|P|$ denotes the length of $P$. 
A path $P=(e_1, \cdots ,e_{\ell})$ is said to have a {\em backtracking} 
if $ e_{r+1} =e_r^{-1} $ for some $r=1,\ldots,\ell-1$, and is said to have a {\em tail} 
if $ e_{\ell} =e_1^{-1} $. 
A path $P$ is called a {\em cycle} if $t(P)=o(P)$. 
Two cycles $C_1 =(e_1, \cdots ,e_{\ell})$ and 
$C_2 =(f_1, \cdots ,f_{\ell})$ are said to be {\em equivalent} if there exists 
$s$ such that $f_r =e_{r+s}$ for all $r$ where indices are treated modulo $\ell$. 
$[C]$ denotes the equivalence class to which $C$ belongs and 
$B^r$ is called the $r$-th {\em power} of a cycle $B$, namely, $B^r$ is obtained by going $r$ times around a cycle $B$. 
% Such a cycle is called a {\em power} of $B$. 
A cycle $C$ is said to be {\em reduced} if 
$C$ has neither backtracking nor tail and to be {\em prime} if it is not a power of 
a strictly smaller cycle. 
We denote by $\mathscr{PR}$ the set of all equivalence classes of prime reduced cycles of $G$. 

The {\em Ihara zeta function} of a graph $G$ is defined by 
\[
{\bf Z} (G, t)= {\bf Z}_G (t)= \prod_{[C]{\in}\mathscr{PR}} (1- t^{ \mid C \mid } )^{-1} ,
\]
where $t$ is a complex variable with $|t|$ sufficiently small. 

In order to explain determinant expressions of ${\bf Z} (G, t)$, we introduce 
two matrices ${\bf B}=({\bf B}_{ef})_{e,f \in D(G)},\,
{\bf J}_0=( {\bf J}_{ef} )_{e,f \in D(G)}{\;\in\;}\Mat(2m,\CM)$ as follows: 
\[
{\bf B}_{ef} =\left\{
\begin{array}{ll}
1 & \mbox{if $t(e)=o(f)$, } \\
0 & \mbox{otherwise,}
\end{array}
\right.
{\bf J}_{ef} =\left\{
\begin{array}{ll}
1 & \mbox{if $f= e^{-1} $, } \\
0 & \mbox{otherwise.}
\end{array}
\right.
\]
The matrix ${\bf B} - {\bf J}_0 $ is called the {\em edge (adjacency) matrix} of $G$.

\begin{thm}[Hashimoto \cite{Hashimoto1989}; Bass \cite{Bass1992}]
The reciprocal of the Ihara zeta function of $G$ is given by 
\[
{\bf Z} (G, t)^{-1} =\det ( {\bf I}_{2m} -t ( {\bf B} - {\bf J}_0 ))
=(1- t^2 )^{r-1} \det ( {\bf I}_n -t {\bf A}+ 
t^2 ({\bf D} -{\bf I}_n )), 
\]
where $r$ and ${\bf A}$ are the Betti number and the adjacency matrix 
of $G$, respectively, and ${\bf D} =({\bf D}_{uv})_{u,v{\in}V(G)}$ is the diagonal matrix 
with ${\bf D}_{uu} = \deg u $ for all $u{\;\in\;}V(G)$. 
\end{thm}

Let ${\bf W}=({\bf W}_{uv})_{u,v{\in}V(G)}{\;\in\;}\Mat(n,\CM)$ such that 
${\bf W}_{uv}{\;\neq\;}0$ if $(u,v){\;\notin\;}D(G)$. 
${\bf W}$ is called a {\em weighted matrix} of $G$.
Put $w(e)= {\bf W}_{uv}$ if $e= (u,v){\;\in\;}D(G)$. 
For a weighted matrix ${\bf W}$ of $G$, the complex matrix 
${\bf B}_w=( {\bf B}^{(w)}_{ef} )_{e,f \in D(G)}{\;\in\;}\Mat(2m,\CM)$ is defined as follows: 
\begin{equation*}
{\bf B}^{(w)}_{ef} =\left\{
\begin{array}{ll}
w(f) & \mbox{if $t(e)=o(f)$, } \\
0 & \mbox{otherwise.}
\end{array}
\right.
\end{equation*}
In \cite{Sato}, Sato defined the second weighted zeta function of a graph. 
The {\em second weighted zeta function} of $G$ is defined by 
\[
{\bf Z}_1 (G,w,t)= \det ( {\bf I}_{2m} -t ( {\bf B}_w - {\bf J}_0 ) )^{-1} . 
\]
\begin{rem}
Precisely, Sato defined a new Bartholdi zeta function ${\bf Z}_1 (G,w,u,t)$ as follows: 
\[
{\bf Z}_1 (G,w,u,t)=\det({\bf I}_n-t({\bf B}_w-(1-u){\bf J}_0),
\]
which reduces to ${\bf Z}_1 (G,w,t)$ with the specialization $u=0$. 
${\bf Z}_1 (G,w,t)$ can also be viewed as a natural generalization of the zeta function 
of the edge-indexed graph defined by Bass \cite{Bass1992}. 
\end{rem}
If $w(e)=1$ for all $e \in D(G)$, then the second weighted zeta function of $G$ 
coincides with the Ihara zeta function of $G$.

\begin{thm}[Sato \cite{Sato}]\label{SatoThm}
Let $G$ be a connected graph and ${\bf W}$ a weighted matrix of $G$. 
Then the reciprocal of the second weighted zeta function of $G$ is given by 
\begin{equation}\label{EqnDetFormula}
{\bf Z}_1 (G,w,t )^{-1} =(1- t^2 )^{m-n} 
\det ({\bf I}_n -t {\bf W}+ t^2 ( {\bf D}_w - {\bf I}_n )) , 
\end{equation}
where $n=|V(G)|$, $m=|E(G)|$ and 
${\bf D}_w =({\bf D}^{(w)}_{uv})_{u,v{\in}V(G)}$ is the diagonal matrix 
with ${\bf D}^{(w)}_{uu} = \displaystyle \sum_{e:o(e)=u}w(e)$ for all 
$u{\;\in\;}V(G)$. 
\end{thm}

Taking transposes of both sides of (\ref{EqnDetFormula}), we obtain the following equation: 

\begin{cor}\label{SatoThmTransposedType}
\[
\det ( {\bf I}_{2m} -t ( {}^T\!{\bf B}_w - {\bf J}_0 ) ) =(1- t^2 )^{m-n} 
\det ({\bf I}_n -t {}^T\!{\bf W}+ t^2 ( {\bf D}_w - {\bf I}_n )). 
\]
\end{cor}

In order to apply zeta functions to obtain the spectral mapping property for 
the quaternionic quantum walk on a graph, we will generalize the determinant formula 
of transposed type in the case of quaternions in the next section.

\section{A quaternionization of the determinant formula for the second weighted zeta function}
\label{sec:QuatDetFormula}
We shall give a generalization of Corollary \ref{SatoThmTransposedType}. 
Hereafter, we assume that $G$ has no multiple edges and has at most one loop at each vertex, 
unless otherwise stated. Let $L(G)$ be the set of loops in $G$ and 
$|V(G)|=n,\;|E(G)|=m=m_0+m_1$, where $m_1=|L(G)|$. 
By the definition of arc, $e^{-1}=e$ if and only if $e$ is a loop. 
Hence it follows that $|D(G)|=m'=2m_0+m_1$. 
For a vertex $u$, $d_u=\deg u$, $d^-_u=\deg^- u$ and $d^+_u=\deg^+ u$ denote 
the {\it degree} of $u$, the {\it indegree} of $u$ and the {\it outdegree} of $u$, 
respectively. 

Let $a$ and $b$ be two maps from $D(G)$ to $\HM$ 
and $a(e)=a(e)^S+ja(e)^P$ and $b(e)=b(e)^S+jb(e)^P$ their symplectic decompositions. 
We define ${\bf K}=({\bf K}_{ev})_{e{\in}D(G),v{\in}V(G)}$ and 
${\bf L}=({\bf L}_{ev})_{e{\in}D(G),v{\in}V(G)}$ to be $m'{\times}n$ matrices as follows: 
\begin{equation}\label{DefMatrixKandL}
{\bf K}_{ev}=
\begin{cases}
a(e) & \text{if $o(e)=v$,}\\
0 & \text{otherwise,}
\end{cases}\quad
{\bf L}_{ev}=
\begin{cases}
b(e) & \text{if $t(e)=v$,}\\
0 & \text{otherwise,}
\end{cases}
\end{equation}
where column index and row index are ordered by fixed sequences 
$v_1,{\ldots},v_n$ and $e_1,{\ldots},e_{m'}$ such that $e_{2r}=e_{2r-1}^{-1}$ for 
$r=1,{\ldots},2m_0$ and $e_r$ is a loop for $r=2m_0+1,{\ldots},m'$, 
respectively. 
${\bf K}{\bf L}^*=(({\bf K}{\bf L}^*)_{ef})_{e,f{\in}D(G)}$ is the $m'{\times}m'$ 
matrix given by the following equation: 
\begin{equation}\label{DefMatrixKL*}
({\bf K}{\bf L}^*)_{ef}=
\begin{cases}
a(e)b(f)^* & \text{if $t(f)=o(e)$,}\\
0 & \text{otherwise.}
\end{cases}
\end{equation}
In order to obtain a quaternionization of Corollary \ref{SatoThmTransposedType}, 
we define two more matrices 
$\tilde{\bf W}=(\tilde{\bf W}_{uv})_{u,v{\in}V(G)}$, 
$\tilde{\bf D}=(\tilde{\bf D}_{uv})_{u,v{\in}V(G)}$ as follows: 
\begin{equation}\label{DefMatrixTildeW}
\begin{split}
\tilde{\bf W}_{uv}&=
\begin{cases}
b((v,u))^*a((v,u)) & \text{if $(v,u){\;\in\;}D(G)$,}\\
0 & \text{otherwise,}
\end{cases}\\
\tilde{\bf D}_{uv}&=
\begin{cases}
\displaystyle \sum_{e{\in}D(G),o(e)=u}
b(e^{-1})^*a(e) & \text{if $u=v$,}\\
0 & \text{otherwise,}
\end{cases}
\end{split}
\end{equation}
We call ${\bf \tilde{W}}$ the {\it doubly weighted matrix}. 
We can check by calculations that 
$\tilde{\bf W}={\bf L}^*{\bf K}$ and ${\bf L}^*{\bf J}_0{\bf K}=\tilde{\bf D}$. 
We also notice that ${\bf J}_0$ is the block diagonal matrix of the form: 
\[
{\bf J}_0=\begin{pmatrix}
{\bf J}_2& & & \\
 &\ddots& & \\
 & &{\bf J}_2& \\
 & & &{\bf I}_{m_1}
\end{pmatrix},\text{ where }{\bf J}_2=\begin{pmatrix}0&1\\1&0\end{pmatrix} 
\text{ places $m_0$ times along the diagonal.}
\]
Before stating the determinant formula, we mention a useful fact on determinant. 

\begin{lem}\label{LemDeterminant}
Let ${\bf A}$ and ${\bf B}$ be an $m{\times}n$ and an $n{\times}m$ matrices, 
respectively. Then for every complex number $\alpha$, the following holds:  
\begin{equation*}
\det (\alpha{\bf I}_m-{\bf AB})\alpha^n=\alpha^m\det (\alpha{\bf I}_n-{\bf BA}).
\end{equation*}
\end{lem}

\begin{proof}
Since 
\begin{equation*}
\begin{split}
\begin{pmatrix}\beta{\bf I}_m & {\bf A} \\ 
{\bf B} & \beta{\bf I}_n\end{pmatrix}
\begin{pmatrix}\beta{\bf I}_m & {\bf O} \\ 
-{\bf B} & \beta{\bf I}_n\end{pmatrix}
&=\begin{pmatrix}\beta^2{\bf I}_m-{\bf AB} & \beta{\bf A} \\ 
{\bf O} & \beta^2{\bf I}_n\end{pmatrix},\\
\begin{pmatrix}\beta{\bf I}_m & {\bf O} \\ 
-{\bf B} & \beta{\bf I}_n\end{pmatrix}
\begin{pmatrix}\beta{\bf I}_m & {\bf A} \\ 
{\bf B} & \beta{\bf I}_n\end{pmatrix}
&=\begin{pmatrix}\beta^2{\bf I}_m & \beta{\bf A} \\ 
{\bf O} & \beta^2{\bf I}_n-{\bf BA}\end{pmatrix},
\end{split}
\end{equation*}
one can see by letting $\alpha=\beta^2$ that
\begin{equation*}
\begin{split}
&\det (\alpha{\bf I}_m-{\bf AB})\alpha^n \\
&=\det \begin{pmatrix}\beta^2{\bf I}_m-{\bf AB} & \beta{\bf A} \\ 
{\bf O} & \beta^2{\bf I}_n\end{pmatrix}
=\det \begin{pmatrix}\beta{\bf I}_m & {\bf A} \\ 
{\bf B} & \beta{\bf I}_n\end{pmatrix}
\det \begin{pmatrix}\beta{\bf I}_m & {\bf O} \\ 
-{\bf B} & \beta{\bf I}_n\end{pmatrix}\\
&=\det (\begin{pmatrix}\beta{\bf I}_m & {\bf O} \\ 
-{\bf B} & \beta{\bf I}_n\end{pmatrix}
\det \begin{pmatrix}\beta{\bf I}_m & {\bf A} \\ 
{\bf B} & \beta{\bf I}_n\end{pmatrix}
=\det \begin{pmatrix}\beta^2{\bf I}_m & \beta{\bf A} \\ 
{\bf O} & \beta^2{\bf I}_n-{\bf BA}\end{pmatrix}\\
&=\alpha^m\det (\alpha{\bf I}_n-{\bf BA})
\end{split}
\end{equation*}
\end{proof}
 
Now we state a quaternionization of the determinant formula for 
the second weighted zeta function of a graph with at most one loop at each vertex. 

\begin{thm}\label{QuatDetFormula}
Let $t$ be a complex variable. Then
\begin{equation*}
\begin{split}
&\det({\bf I}_{2m'}-t\psi({\bf K}{\bf L}^*-{\bf J}_0))\\
&=(1-t^2)^{2m_0-2n}(1+t)^{2m_1}
\det({\bf I}_{2n}-t\psi(\tilde{\bf W})+t^2(\psi(\tilde{\bf D})-{\bf I}_{2n}))
\end{split}
\end{equation*}
\end{thm}

\begin{proof} 
Using Lemma \ref{LemDeterminant}, we have the following sequence of equations: 
\begin{equation}\label{EqnSecDeterminants}
\begin{split}
&\det({\bf I}_{2m'}-t\psi({\bf K}{\bf L}^*-{\bf J}_0))\det({\bf I}_{2m'}-t\psi({\bf J}_0))(1-t^2)^{2n}\\
&=\det({\bf I}_{2m'}-t\psi({\bf K})\psi({\bf L})^*+t\psi({\bf J}_0))\det({\bf I}_{2m'}-t\psi({\bf J}_0))(1-t^2)^{2n}\\
&=\det((1-t^2){\bf I}_{2m'}-t\psi({\bf K})\psi({\bf L})^*({\bf I}_{2m'}-t\psi({\bf J}_0)))(1-t^2)^{2n}\\
&=\det((1-t^2){\bf I}_{2n}-t\psi({\bf L})^*({\bf I}_{2m'}-t\psi({\bf J}_0))\psi({\bf K}))(1-t^2)^{2m'}\\
&=\det((1-t^2){\bf I}_{2n}-t\psi({\bf L})^*\psi({\bf K})+t^2\psi({\bf L})^*\psi({\bf J}_0)\psi({\bf K}))(1-t^2)^{2m'}\\
&=\det((1-t^2){\bf I}_{2n}-t\psi(\tilde{\bf W})+t^2\psi(\tilde{\bf D}))(1-t^2)^{2m'}.
\end{split}
\end{equation}
Since the following holds: 
\begin{equation*}
\det({\bf I}_{2m'}-t\psi({\bf J}_0))=(1-t^2)^{2m_0}(1-t)^{2m_1},
\end{equation*}
for every $t{\;\in\;}\CM$, we have 
\begin{equation}\label{EqnSecDeterminants2}
\begin{split}
&(1-t^2)^{2m_0+2n}(1-t)^{2m_1}\det({\bf I}_{2m'}-t\psi({\bf K}{\bf L}^*-{\bf J}_0))\\
&=(1-t^2)^{4m_0+2m_1}\det({\bf I}_{2n}-t\psi(\tilde{\bf W})+t^2(\psi(\tilde{\bf D})-{\bf I}_{2n})).
\end{split}
\end{equation}
The left hand side and the right hand side of (\ref{EqnSecDeterminants2}) are equal as polynomials and thus 
the assertion follows. 
\end{proof}

\section{A quaternionization of quantum walks of Szegedy type and right spectral problem}
\label{sec:TransMat}
In this section, we give a quaternionization of quantum walks of Szegedy type on graphs 
and derive their spectral mapping properties. 
An important example of the quantum walk on a graph 
is the Grover walk. The Grover walk originates from a quantum search algorithm introduced 
by Grover \cite{Grover1996}. 
After Grover's approach, several quantum search algorithms, for example 
\cite{Ambainis2003b,Szegedy2004}, was proposed. These are strongly related to 
quantum walks on graphs, and have been investigated for more than a decade. 
In this article, we focus on the discrete-time quantum walk on a graph based on 
the quantum search algorithm proposed by Szegedy \cite{Szegedy2004}. 

We give a definition of discrete-time quantum walks on graphs in a similar manner 
as in \cite{Ambainis2003}. 
Other formulations of discrete-time quantum walks on graphs can be seen in, 
for example, \cite{Watrous2001,RenETAL}. 
Let $\mathscr{H}=\oplus_{e{\in}D(G)}{\CM}|e\rangle$ be the finite dimensional Hilbert 
space spanned by arcs of $G$. 
One step of quantum walk consists of the following two operations ${\bf C}$ and ${\bf S}$ 
as follows:
\begin{enumerate}
\renewcommand{\labelenumi}{\rm \arabic{enumi}.}
\item
For each $u{\;\in\;}V$, we perform a unitary transformation ${\bf C}_u$ 
on the states $|f\rangle$ that satisfy $t(f)=u$ and put ${\bf C}$ the direct sum of 
these transformations: ${\bf C}=\oplus_{u{\in}V}{\bf C}_u$. 
${\bf C}$ is called the {\it coin} operator. 
\item
For all $e{\;\in\;}D(G)$, we perform the {\it shift} operator ${\bf S}$ defined by 
${\bf S}|e{\rangle}=|e^{-1}{\rangle}$. 
\end{enumerate} 
The transition matrix ${\bf U}$ of a discrete-time quantum walk consists of 
the two consecutive operations ${\bf U}={\bf SC}$. 
Let $p$ be a map from $D(G)$ to $(0,1]$ which satisfies 
$\sum_{e{\in}D(G),o(e)=u}p(e)=1$ for every $u{\;\in\;}V(G)$ and 
$|\phi_u\rangle$ the state vector defined by 
\begin{equation}\label{DefPhi_u}
|\phi_u\rangle=\sum_{f{\in}D(G),t(f)=u}\sqrt{p(f^{-1})}|f\rangle.
\end{equation} 
Then the coin operator ${\bf C}_u$ of the Szegedy walk on $G$ is given by 
\begin{equation*}
{\bf C}_u=2|\phi_u{\rangle}{\langle}\phi_u|-{\bf I}_{d^-_u}.
\end{equation*}
One can immediately check that the transition matrix 
${\bf U}^{\rm Sze}=({\bf U}^{\rm Sze}_{ef})_{e,f \in D(G)}$ 
of the Szegedy walk is given by 
\[
{\bf U}^{\rm Sze}_{ef} =\left\{
\begin{array}{ll}
2\sqrt{p(e)p(f^{-1})} & \mbox{if $t(f)=o(e)$ and $f{\;\neq\;}e^{-1}$, } \\
2p(e) -1 & \mbox{if $f= e^{-1} $, } \\
0 & \mbox{otherwise,}
\end{array}
\right. 
\]
We call ${\bf U}^{\rm Sze}$ the {\it Szegedy matrix}. 
Higuchi et al. \cite{HKSS2013} generalized the Szegedy walk on a graph 
and obtained the spectral mapping theorem for the generalized Szegedy walk. 

Now, we extend the Szegedy walk on a graph to the case of quaternions. 
A background of quaternionic quantum walks is based on quaternionic quantum mechanics. 
In a similar way to \cite{KMS2016}, we briefly explain the general framework of quaternionic 
quantum mechanics to the minimum extent necessary for the definition of 
quaternionic quantum walks. 
For a more detailed and complete exposition of quaternionic quantum mechanics, see \cite{Adler1995}. 
The states of quaternionic quantum mechanics is described by vectors of 
a quaternionic Hilbert space 
$\mathscr{H}_{\HM}$ 
which is defined by the following axioms: 
\begin{enumerate}
\renewcommand{\labelenumi}{\rm (\Roman{enumi})}
\item
$\mathscr{H}_{\HM}$ is a right $\HM$-vector space. 
\item
There exists a scalar product, that is a map 
$(\ ,\ ) : \mathscr{H}_{\HM}{\times}\mathscr{H}_{\HM}{\longrightarrow}\HM$ 
with the properties: 
\begin{equation*}
\begin{split}
({\bf f},\,{\bf g})^*&=({\bf g},\,{\bf f}),\\
({\bf f},\,{\bf f})& \text{ is nonnegative real number, and } 
({\bf f},\,{\bf f})=0{\;\Leftrightarrow\;}{\bf f}=0,\\
({\bf f},\,{\bf g}+{\bf h})&=({\bf f},\,{\bf g})+({\bf f},\,{\bf h})\\
({\bf f},\,{\bf g}{\lambda})&=({\bf f},\,{\bf g})\lambda \quad \text{for every $\lambda{\in}\HM$}.
\end{split}
\end{equation*}
$\|{\bf f}\|$ denotes $\sqrt{({\bf f},\,{\bf f})}$ and is called the {\it norm} of ${\bf f}$. 
\item
$\mathscr{H}_{\HM}$ is separable and complete. 
\end{enumerate}
Ket states are denoted by $|{\bf f}{\rangle}$ that satisfy 
$|{\bf f}{\lambda}{\rangle}=|{\bf f}{\rangle}\lambda$ for every $\lambda{\in}\HM$. 
Bra states ${\langle}{\bf f}|$ are defined as their adjoints in a matrix sense so that 
${\langle}{\bf f}|=|{\bf f}{\rangle}^{\!\dagger}$ where $|{\bf f}{\rangle}^{\!\dagger}$ denotes 
the adjoint of $|{\bf f}{\rangle}$. 
It follows that ${\langle}{\bf f}{\lambda}|={\lambda}^*{\langle}{\bf f}|$ and 
the scalar product can be presented as $({\bf f},\,{\bf g})={\langle}{\bf f}|{\bf g}{\rangle}$. 
${\langle}{\bf f}|{\bf g}{\rangle}$ is called the {\it inner product} or {\it probability amplitude}. 
If $\mathscr{H}_{\HM}$ is $N$-dimensional, 
then $|{\bf f}{\rangle}$ and ${\langle}{\bf f}|$ are presented as follows: 
\begin{equation*}
|{\bf f}{\rangle}={}^{T\!}(f_1\,f_2\,{\cdots}\,f_N),\quad {\langle}{\bf f}|=({f_1}^*\,{f_2}^*\,{\cdots}\,{f_N}^*) \quad 
(f_1,f_2,{\cdots},f_N{\;\in\;}\HM).
\end{equation*}
Using the coordinates, the scalar product can be written as  
\begin{equation*}
{\langle}{\bf f}|{\bf g}{\rangle}=\sum_{\ell=1}^N {f_{\ell}}^*g_{\ell}. 
\end{equation*}
For a probability amplitude ${\langle}{\bf f}|{\bf g}{\rangle}$, one associates a probability: 
\begin{equation*}
P_{{\bf f}{\bf g}}=|{\langle}{\bf f}|{\bf g}{\rangle}|^2.
\end{equation*}
For any $\omega{\;\in\;}\HM$ such that $|\omega|=1$, 
$|{\langle}{\bf f}|{\bf g}{\omega}{\rangle}|^2=|{\langle}{\bf f}|{\bf g}{\rangle}|^2$. 
Hence $P_{{\bf f}({\bf g}{\omega})}=P_{{\bf f}{\bf g}}$ for all ${\bf f}{\;\in\;}\mathscr{H}_{\HM}$. 
Thus physical states bijectively corresponds to sets of unit rays of $\mathscr{H}_{\HM}$ of the form: 
\begin{equation*}
\{|{\bf g}{\omega}{\rangle}\,|\,|{\omega}|=1\}.
\end{equation*} 
Operators ${\bf U}$ on $\mathscr{H}_{\HM}$ are assumed to act from the left as in ${\bf U}|{\bf f}{\rangle}$ and satisfy 
\begin{equation*}
{\bf U}|{\bf f}+{\bf g}{\rangle}={\bf U}|{\bf f}{\rangle}+{\bf U}|{\bf g}{\rangle}, \quad 
{\bf U}(|{\bf f}{\rangle}{\lambda})=({\bf U}|{\bf f}{\rangle}){\lambda},
\end{equation*}
for all $\lambda{\;\in\;}\HM$. 
For any operator ${\bf U}$, the {\it adjoint operator} ${\bf U}^{\dagger}$ is defined by 
$({\bf f},\,{\bf U}{\bf g})=({\bf U}^{\dagger}{\bf f},\,{\bf g})$ 
for arbitrary ${\bf f},\,{\bf g}$ in suitable domains. 
An operator ${\bf U}$ is said to be {\it self-adjoint} if 
${\bf U}^{\dagger}={\bf U}$ holds and {\it unitary} if 
${\bf U}{\bf U}^{\dagger}={\bf U}^{\dagger}{\bf U}={\bf 1}$ holds. 
If $\mathscr{H}_{\HM}$ is finite dimensional, unitary operators are 
unitary as quaternionic matrices. 
In order to explain the dynamics of quaternionic quantum mechanics, we explain the {\it time evolution operator}. 
We assume that time evolution preserves transition probabilities. Precisely, 
for arbitrary states $|{\bf f}(t){\rangle},|{\bf g}(t){\rangle}$ at time $t$ and 
$|{\bf f}(t'){\rangle},|{\bf g}(t'){\rangle}$ at time $t'=t+{\delta}t$, 
$|{\langle}{\bf f}(t)|{\bf g}(t){\rangle}|=|{\langle}{\bf f}(t')|{\bf g}(t'){\rangle}|$ holds. 
Choosing appropriate quaternionic phases, we obtain that 
there exists a unitary operator ${\bf U}(t',t)$ such that 
\begin{equation}\label{eqnTimeEvolution}
|{\bf f}(t'){\rangle}={\bf U}(t',t)|{\bf f}(t){\rangle}.
\end{equation}
${\bf U}(t',t)$ is called the {\it time evolution operator} and 
(\ref{eqnTimeEvolution}) is called the {\it time evolution equation}. 
If $\mathscr{H}_{\HM}$ is $N$-dimensional, then ${\bf U}(t',t)$ can be represented as 
an $N{\times}N$ quaternionic unitary matrix. 

Let us give the definition of discrete-time quaternionic quantum walk. 
A discrete-time quaternionic quantum walk is a quantum process on a graph whose 
state vector, whose entries are quaternions, 
is governed by a quaternionic unitary matrix called the transition matrix 
(or time evolution matrix). 
Let $q$ be a map from $D(G)$ to $\HM^*=\HM-\{0\}$.
We define the state space of the walk to be the right Hilbert space 
$\mathscr{H}_{\HM}=\oplus_{e{\in}D(G)}|e\rangle{\HM}$ over $\HM$ generated by the 
orthonormal basis $\{\,|e\rangle\,|\,e{\;\in\;}D(G)\,\}$. 
We define the coin operator ${\bf C}_u$ by replacing $\sqrt{p(f^{-1})}$ with $q(f^{-1})$ 
in (\ref{DefPhi_u}). 
Since a state vector is multiplied by a scalar on the right, 
$|\phi_u{\rangle}$ turns into 
\begin{equation*}
|\phi_u\rangle=\sum_{f{\in}D(G),t(f)=u}|f{\rangle}q(f^{-1}).
\end{equation*} 
One can readily check that the {\em transition matrix} 
${\bf U} =( {\bf U}_{ef} )_{e,f \in D(G)} $ 
of quaternionic Szegedy walk on $G$ is given by 
\begin{equation*}
{\bf U}_{ef} =\left\{
\begin{array}{ll}
2q(e)q(f^{-1})^* & \mbox{if $t(f)=o(e)$ and $f \neq e^{-1} $, } \\
2|q(e)|^2 -1 & \mbox{if $f= e^{-1} $, } \\
0 & \mbox{otherwise.}
\end{array}
\right.
\end{equation*}
${\bf U}$ can be viewed as the time evolution operator of 
a discrete-time quaternionic quantum system. 
The quaternionic unitary condition on ${\bf U}$ is given by the following: 

\begin{thm}\label{ThmUnitaryCondition}
${\bf U}$ is quaternionic unitary if and only if 
\[ 
\sum_{g{\in}D(G),o(g)=u}|q(g)|^2=1 
\]
for each $u{\;\in\;}V(G)$.
\end{thm}

\begin{proof} 
Observing that 
\[
({\bf U}{\bf U}^*)_{ef}=\sum_{\substack{g{\in}D(G)\\t(g)=o(e)=o(f)}}{\bf U}_{eg}({\bf U}_{fg})^*,
\]
one can readily check by direct calculations that 
$({\bf U}{\bf U}^*)_{ee}=1$ is equivalent to 
\begin{equation}\label{EqnUnitaryCond1}
|q(e)|^2(\sum_{\substack{g{\in}D(G)\\o(g)=o(e)}}|q(g)|^2-1)=0,
\end{equation}
and that $({\bf U}{\bf U}^*)_{ef}=0$ for $e{\;\neq\;}f$ with $o(e)=o(f)$ is equivalent to 
\begin{equation}\label{EqnUnitaryCond2}
q(e)q(f)^*(\sum_{\substack{g{\in}D(G)\\o(g)=o(e)=o(f)}}|q(g)|^2-1)=0.
\end{equation}
Moreover, if $o(e){\;\neq\;}o(f)$ then $({\bf U}{\bf U}^*)_{ef}=0$ holds directly 
from the definition of ${\bf U}$. 
Since $q(e){\;\neq\;}0$ for every $e{\;\in\;}D(G)$, 
$\sum_{g{\in}D(G)o(g)=u}|q(g)|^2=1$ for every $u{\;\in\;}V(G)$ is equivalent to 
(\ref{EqnUnitaryCond1}) and (\ref{EqnUnitaryCond2}). 
\end{proof}

We put $a(e)=\sqrt{2}q(e)$ and $b(e)=\sqrt{2}q(e^{-1})$ in (\ref{DefMatrixKandL}). 
Then one can see that ${\bf K}{\bf L}^*-{\bf J}_0={\bf U}$ and that $\tilde{\bf W}$ 
and $\tilde{\bf D}$ in (\ref{DefMatrixTildeW}) turn into: 
\begin{equation}\label{EqnQSzegedyWandDMatrix}
\begin{split}
\tilde{\bf W}_{uv}&=
\begin{cases}
2q((u,v))^*q((v,u)) & \text{if $(v,u){\;\in\;}D(G)$,}\\
0 & \text{otherwise,}
\end{cases}\\
\tilde{\bf D}_{uv}&=
\begin{cases}
\displaystyle \sum_{e{\in}D(G),o(e)=u}
2|q(e)|^2 & \text{if $u=v$,}\\
0 & \text{otherwise.}
\end{cases}
\end{split}
\end{equation}
(\ref{EqnQSzegedyWandDMatrix}) implies that $\tilde{\bf W}$ is quaternionic 
Hermitian, that is $\tilde{\bf W}^*=\tilde{\bf W}$. 
Therefore $\psi(\tilde{\bf W})$ is Hermitian and 
all eigenvalues of $\psi(\tilde{\bf W})$ are real numbers. 
Furthermore if $q$ satisfies Theorem \ref{ThmUnitaryCondition}, then 
$\tilde{\bf D}=2{\bf I}_{n}$ holds and we have from Theorem \ref{QuatDetFormula} that 
\begin{equation}\label{EqnDeterminantformula}
\begin{split}
&\det({\bf I}_{2m'}-t\psi({\bf U}))\\
&=(1-t^2)^{2m_0-2n}(1+t)^{2m_1}
\det({\bf I}_{2n}-t\psi(\tilde{\bf W})+t^2(\psi(\tilde{\bf D})-{\bf I}_{2n}))\\
&=(1-t^2)^{2m_0-2n}(1+t)^{2m_1}
\det({\bf I}_{2n}-t\psi(\tilde{\bf W})+t^2{\bf I}_{2n}).
\end{split}
\end{equation}
Putting $t=1/{\lambda}$, (\ref{EqnDeterminantformula}) yields 
\begin{equation}\label{EqnCharacteristicPoly}
\begin{split}
&\det({\lambda}{\bf I}_{2m'}-\psi({\bf U}))\\
&=(\lambda^2-1)^{2m_0-2n}(\lambda+1)^{2m_1}
\det({\lambda^2}{\bf I}_{2n}-{\lambda}\psi(\tilde{\bf W})+{\bf I}_{2n})\\
&=(\lambda-1)^{2m_0-2n}(\lambda+1)^{2m_0-2n+2m_1}
\det({\lambda^2}{\bf I}_{2n}-{\lambda}\psi(\tilde{\bf W})+{\bf I}_{2n}).
\end{split}
\end{equation}
We denote by $\Spec ({\bf A})$ the multiset of eigenvalues of a complex square matrix 
${\bf A}$ counted with multiplicity. 
Using (\ref{EqnCharacteristicPoly}), we obtain the spectral mapping theorem 
as described below: 

\begin{thm}\label{ThmSpecMapping}
$\Spec(\psi({\bf U}))$ has $2m'$ elements. 
If $G'=(V(G),E(G)-L(G))$ is not a tree, then $4n$ of them are 
\begin{equation}\label{EqnEigenvalueOfPsiU}
\lambda=\dfrac{\mu}{2}{\pm}i\sqrt{1-\Big{(}\dfrac{\mu}{2}\Big{)}^2},
\end{equation}
where $\mu{\;\in\;}\Spec(\psi(\tilde{\bf W}))$. 
The remaining $2m'-4n$ eigenvalues are equal to $1$ with multiplicity $2m_0-2n$ and to 
$-1$ with multiplicity $2m_0+2m_1-2n$. 
If $G'$ is a tree and $G$ has no loop, then 
\begin{equation*}
\begin{split}
\Spec(\psi({\bf U}))
=\Big{\{}\,\dfrac{\mu}{2}{\pm}i\sqrt{1-\Big{(}\dfrac{\mu}{2}\Big{)}^2\,}\;
\Bigl|\; \mu {\in} \Spec(\psi(\tilde{\bf W}))\Big{\}}-\{1,1,-1,-1\},
\end{split}
\end{equation*}
and if $G'$ is a tree and $G$ has at least one loop, then 
% ${\Spec}(\psi({\bf U}))$ is the union of 
\begin{equation*}
\begin{split}
\Spec(\psi ({\bf U}))
=\bigg{[}\Big{\{}\,\dfrac{\mu}{2}{\pm}i\sqrt{1-\Big{(}\dfrac{\mu}{2}\Big{)}^2\,}\;
\Bigl|\; \mu {\in} \Spec(\psi(\tilde{\bf W}))\Big{\}}-\{1,1\}\bigg{]}{\cup}
\{\underbrace{-1,\ldots,-1}_{2m_1-2}\}.
\end{split}
\end{equation*}
In all cases, each eigenvalue $\mu{\;\in\;}\Spec(\psi(\tilde{\bf W}))$ takes values 
between $-2$ and $2$. 
Elements of $\Spec(\psi ({\bf U}))$ occur in complex conjugate pairs such as 
$\lambda_1,\overline{\lambda_1},{\cdots},\lambda_{2m},\overline{\lambda_{2m}}$, 
and the set of right eigenvalues of ${\bf U}$ is given by 
$\sigma_r({\bf U})=\bigcup_{r=1}^{2m}\lambda_{r}^{\HM^*}$. 
\end{thm}

\begin{proof} Let $\mu_1,{\cdots},\mu_{2n}$ be eigenvalues of $\psi(\tilde{\bf W})$. 
Then it follows from (\ref{EqnCharacteristicPoly}) that 
\[
\det ( \lambda {\bf I}_{2m'} - \psi({\bf U}) )= 
(\lambda-1)^{2m_0-2n}(\lambda+1)^{2m_0-2n+2m_1}\prod_{r=1}^{2n}(\lambda {}^2-\mu_r\lambda+1). 
\]
Solving $\lambda {}^2 -\mu_r\lambda+1 =0$, we obtain complex eigenvalues of 
$\psi({\bf U})$ as follows: 
\begin{equation}\label{EqnEigenvaluesFormula}
\lambda = \dfrac{\mu_r \pm \sqrt{\mu_r^2-4}}{2}.
\end{equation}
Since $\psi({\bf U})$ is unitary, $|\lambda|=1$ for every eigenvalue $\lambda$ of 
$\psi({\bf U})$, and from (\ref{EqnEigenvaluesFormula}) we have 
$\lambda=1$ (resp. $\lambda=-1$) if and only if $\mu_r=2$ (resp. $\mu_r=-2$). 
It follows that eigenvalues of $\psi(\tilde{\bf W})$ take values between $-2$ and $2$. 
Hence (\ref{EqnEigenvaluesFormula}) can be reformulated as  
\[
\lambda=\dfrac{\mu_r}{2}{\pm}i\sqrt{1-\Big{(}\dfrac{\mu_r}{2}\Big{)}^2}
\]
whose real part and imaginary part are $\dfrac{\mu_r}{2}$ and 
${\pm}\sqrt{1-\Big{(}\dfrac{\mu_r}{2}\Big{)}^2}$ respectively. 
If $G'$ is not a tree, then $m_0{\;\geq\;}n$ and hence the statement clearly holds. 
If $G'$ is a tree, then $m_0=n-1$. Since $\det ( \lambda {\bf I}_{2m'} - \psi({\bf U}) )$ is 
a polynomial of $\lambda$, $\prod_{r=1}^{2n}(\lambda^2-\mu_r\lambda+1)$ must 
have the factor $(\lambda{}-1)^2$ and, in addition, the factor $(\lambda{}+1)^2$ if 
$G$ has no loops. Thus the statement holds for trees. 
Consequently, $\sigma_r({\bf U})$ is obtained directly from Theorem \ref{QuatEigen}. 
\end{proof}

Let ${\bf T}=\tilde{\bf W}/2$ which is given by 
\begin{equation*}\label{EqnQSzegedyTMatrix}
{\bf T}_{uv}=
\begin{cases}
q((u,v))^*q((v,u)) & \text{if $(v,u){\;\in\;}D(G)$,}\\
0 & \text{otherwise.}
\end{cases}
\end{equation*}
Then $\Spec(\psi({\bf T}))=\{\mu/2\;|\;\mu{\;\in\;}\Spec(\psi({\tilde{\bf W}}))\}$, and 
we can describe the elements of the form (\ref{EqnEigenvalueOfPsiU}) in $\Spec(\psi({\bf U}))$ 
slightly simpler as 
\[
\lambda=\nu{\pm}i\sqrt{1-\nu^2},
\]
where $\nu{\;\in\;}\Spec(\psi({\bf T}))$. 

\section{Eigenvectors of ${\bf U}$ derived from eigenvectors of $\tilde{\bf W}$}
\label{sec:Eigenvectors}
In this section, we show the way to derive right eigenvectors of ${\bf U}$ corresponding to 
complex right eigenvalues derived from those of $\tilde{\bf W}$. 
We mention that eigenspaces of Szegedy walk are discussed in \cite{HKSS2014} 
at length in the complex case. 
If $\lambda$ is an eigenvalue of $\psi({\bf U})$ obtained by solving $\lambda^2-\mu\lambda+1=0$ 
with $\mu{\;\in\;}\Spec(\psi(\tilde{\bf W}))$, then $\lambda$ and $\mu$ satisfy 
\begin{equation}\label{EqnRelationBetweenLambdaAndMu}
\mu=\lambda+\dfrac{1}{\lambda}.
\end{equation}
Let $\mathscr{K}_{\HM}=\oplus_{v{\in}V(G)}|v{\rangle}\HM$ be the finite dimensional 
right Hilbert space over $\HM$ spanned by vertices of $G$. 
We can identify $\tilde{\bf W}$ with the linear transformation 
with respect to the orthonormal basis $|v_1{\rangle},\ldots,|v_n{\rangle}$ defined as follows: 
\[
\tilde{\bf W}|v_r{\rangle}=\sum_{s=1}^n|v_s{\rangle}\tilde{\bf W}_{sr}
=\sum_{s=1}^n|v_s{\rangle}2q((v_s,v_r))^*q((v_r,v_s))
\]
Similarly, we consider ${\bf K}$ and ${\bf L}$ as linear maps from 
$\mathscr{K}_{\HM}$ to $\mathscr{H}_{\HM}$ so that ${\bf K}^*$ and ${\bf L}^*$ can be 
viewed as linear maps from $\mathscr{H}_{\HM}$ to $\mathscr{K}_{\HM}$. 
We shall show the following equations for later use. 

\begin{lem}\label{LemEquationsBetweenJKL}
\begin{enumerate}
\renewcommand{\labelenumi}{\rm (\arabic{enumi})}
\item
${\bf K}^*{\bf K}={\bf L}^*{\bf L}=2{\bf I}_n$.
\item
${\bf J}_0{\bf K}{\bf L}^*={\bf L}{\bf L}^*$. 
\item
${\bf L}^*{\bf J}_0{\bf L}=\tilde{\bf W}$. 
\end{enumerate}
\end{lem}

\begin{proof}
(1) Using Theorem \ref{ThmUnitaryCondition}, it follows that 
\begin{equation*}
\begin{split}
({\bf K}^*{\bf K})_{uv}&=\sum_{e{\in}D(G)}({\bf K}_{eu})^*{\bf K}_{ev}
=\sum_{\substack{e{\in}D(G)\\o(e)=u=v}}2|q(e)|^2=2\delta_{uv},\\
({\bf L}^*{\bf L})_{uv}&=\sum_{e{\in}D(G)}({\bf L}_{eu})^*{\bf L}_{ev}
=\sum_{\substack{e{\in}D(G)\\t(e)=u=v}}2|q(e^{-1})|^2=2\delta_{uv}.
\end{split}
\end{equation*}
(2) By direct computations, we have 
\begin{equation*}
\begin{split}
({\bf J}_0{\bf K}{\bf L}^*)_{ef}
&=\sum_{g{\in}D(G),u{\in}V(G)}{\bf J}_{eg}{\bf K}_{gu}({\bf L}_{fu})^*
=\sum_{u{\in}V(G)}{\bf K}_{e^{-1}u}({\bf L}_{fu})^*\\
&=\sum_{\substack{u{\in}V(G)\\u=o(e^{-1})=t(f)}}2q(e^{-1})q(f^{-1})^*
=\sum_{\substack{u{\in}V(G)\\u=t(e)=t(f)}}2q(e^{-1})q(f^{-1})^*,\\
({\bf L}{\bf L}^*)_{ef}&=\sum_{u{\in}V(G)}{\bf L}_{eu}({\bf L}_{fu})^*
=\sum_{\substack{u{\in}V(G)\\u=t(e)=t(f)}}2q(e^{-1})q(f^{-1})^*.
\end{split}
\end{equation*}
(3) By a direct computation, we have 
\begin{equation*}
\begin{split}
({\bf L}^*{\bf J}_0{\bf L})_{uv}
&=\sum_{e,f{\in}D(G)}({\bf L}^*)_{ue}{\bf J}_{ef}{\bf L}_{fv}
=\sum_{e{\in}D(G)}({\bf L}^*)_{ue}{\bf L}_{e^{-1}v}\\
&=\sum_{\substack{e{\in}D(G)\\t(e)=u,t(e^{-1})=v}}({\bf L}_{eu})^*{\bf L}_{e^{-1}v}\\
&=\begin{cases}2q((u,v))^*q((v,u)) & \text{if $(v,u){\;\in\;}D(G)$,}\\
0 & \text{otherwise.}\end{cases}
\end{split}
\end{equation*}
\end{proof}

Consequently, we can deduce that by using an eigenvector of $\tilde{\bf W}$ 
corresponding to $\mu{\;\in\;}$ $\Spec(\psi(\tilde{\bf W}))$, 
one can express an eigenvector of ${\bf U}$ corresponding to an eigenvalue $\lambda$ 
related to $\mu$ as follows:  

\begin{thm}\label{ThmEigenvectors}
Suppose that ${\bf v}{\;\in\;}\mathscr{K}_{\HM}$ is a right eigenvector 
corresponding to a complex right eigenvalue $\mu$ of $\tilde{\bf W}$, 
namely, $\tilde{\bf W}{\bf v}={\bf v}\mu$. 
Then the following vector in $\mathscr{H}_{\HM}$: 
\[
{\bf e}={\bf J}_0{\bf L}{\bf v}-{\bf L}{\bf v}\dfrac{1}{\lambda}
\]
is a right eigenvector of ${\bf U}$ corresponding to the complex right eigenvalue $\lambda$, 
where $\mu$ and $\lambda$ satisfy (\ref{EqnRelationBetweenLambdaAndMu}). 
\end{thm}

\begin{proof}
Since ${\bf K}{\bf L}^*-{\bf J}_0={\bf U}$, we have 
${\bf U}={\bf J}_0({\bf L}{\bf L}^*-{\bf I})$. Hence using Lemma \ref{LemEquationsBetweenJKL}, 
we obtain 
\begin{equation*}
\begin{split}
{\bf U}{\bf e}&={\bf J}_0({\bf L}{\bf L}^*-{\bf I})
({\bf J}_0{\bf L}{\bf v}-{\bf L}{\bf v}\dfrac{1}{\lambda})\\
&={\bf J}_0{\bf L}{\bf L}^*{\bf J}_0{\bf L}{\bf v}
-{\bf J}_0{\bf L}{\bf L}^*{\bf L}{\bf v}\dfrac{1}{\lambda}
-{\bf J}_0^2{\bf L}{\bf v}+{\bf J}_0{\bf L}{\bf v}\dfrac{1}{\lambda}\\
&={\bf J}_0{\bf L}\tilde{\bf W}{\bf v}-2{\bf J}_0{\bf L}{\bf v}\dfrac{1}{\lambda}
-{\bf L}{\bf v}+{\bf J}_0{\bf L}{\bf v}\dfrac{1}{\lambda}\\
&={\bf J}_0{\bf L}{\bf v}\Big{(}\lambda+\dfrac{1}{\lambda}\Big{)}
-{\bf J}_0{\bf L}{\bf v}\dfrac{1}{\lambda}-{\bf L}{\bf v}\\
&={\bf J}_0{\bf L}{\bf v}\lambda-{\bf L}{\bf v}
=\Big{(}{\bf J}_0{\bf L}{\bf v}-{\bf L}{\bf v}\dfrac{1}{\lambda}\Big{)}\lambda
={\bf e}\lambda
\end{split}
\end{equation*}
\end{proof}

\begin{rem}
{\rm (1)} Since ${\bf \tilde{W}}$ is quaternionic Hermitian, $\mu$ is real in fact. \\
{\rm (2)} ${\bf e}j$ is a right eigenvector corresponding to $\overline{\lambda}$ since 
${\bf v}j$ is a right eigenvector corresponding to $\mu$ and 
\[
{\bf e}j={\bf J}_0{\bf L}{\bf v}j-{\bf L}{\bf v}\dfrac{1}{\lambda}j
={\bf J}_0{\bf L}{\bf v}j-{\bf L}{\bf v}j\dfrac{1}{\overline{\lambda}}. 
\]
\end{rem}

We shall give an example of the spectrum and eigenvectors of the quaternionic Szegedy walk on 
a graph. 

\begin{exm}
Let $G=K_3$, the complete graph with three vertices $v_1,\,v_2,\,v_3$ and 
$e_1=(v_1,v_2),\,e_2=(v_2,v_1)=e_1^{-1},\,e_3=(v_2,v_3),\,e_4=(v_3,v_2)=e_3^{-1},\,
e_5=(v_3,v_1),\,e_6=(v_1,v_3)=e_5^{-1},\,e_7=(v_1,v_1),\,e_8=(v_2,v_2),\,e_9=(v_3,v_3)$. 
Suppose quaternionic weights on $G$ are given by 
\begin{equation*}
\begin{split}
&q(e_1)=\dfrac{i}{\sqrt{3}},\,q(e_2)=-\dfrac{i}{\sqrt{3}},\,
q(e_3)=\dfrac{j}{\sqrt{3}},q(e_4)=-\dfrac{j}{\sqrt{3}},\\
&q(e_5)=\dfrac{k}{\sqrt{3}},\,q(e_6)=-\dfrac{k}{\sqrt{3}},\,
q(e_7)=q(e_8)=q(e_9)=\dfrac{1}{\sqrt{3}}.
\end{split}
\end{equation*}
Then ${\bf U},\,\tilde{\bf W}$ and ${\bf L}$ are given by 
\begin{equation*}
\begin{split}
{\bf U}&=\bordermatrix{     & e_1 & e_2 & e_3 & e_4 & e_5 & e_6 & e_7 & e_8 & e_9 \cr
               e_1 &  0  & -\frac{1}{3} & 0 & 0 & -\frac{2j}{3} & 0 & \frac{2i}{3} & 0 & 0 \cr
               e_2 &  -\frac{1}{3}  & 0 & 0 & \frac{2k}{3} & 0 & 0 & 0 & -\frac{2i}{3} & 0 \cr
               e_3 &  -\frac{2k}{3}  & 0 & 0 & -\frac{1}{3} & 0 & 0 & 0 & \frac{2j}{3} & 0 \cr
               e_4 &  0  & 0 & -\frac{1}{3} & 0 & 0 & \frac{2i}{3} & 0 & 0 & -\frac{2j}{3} \cr
               e_5 &  0  & 0 & -\frac{2i}{3} & 0 & 0 & -\frac{1}{3} & 0 & 0 & \frac{2k}{3} \cr
               e_6 &  0  & \frac{2j}{3} & 0 & 0 & -\frac{1}{3} & 0 & -\frac{2k}{3} & 0 & 0 \cr
               e_7 &  0  & -\frac{2i}{3} & 0 & 0 & \frac{2k}{3} & 0 & -\frac{1}{3} & 0 & 0 \cr
               e_8 &  \frac{2i}{3}  & 0 & 0 & -\frac{2j}{3} & 0 & 0 & 0 & -\frac{1}{3} & 0 \cr
               e_9 &  0  & 0 & \frac{2j}{3} & 0 & 0 & -\frac{2k}{3} & 0 & 0 & -\frac{1}{3} 
            },\\
\tilde{\bf W}&=\bordermatrix{     & v_1 & v_2 & v_3 \cr
               v_1 &  \frac{2}{3}  & -\frac{2}{3} & -\frac{2}{3} \cr
               v_2 & -\frac{2}{3} & \frac{2}{3} & -\frac{2}{3} \cr
               v_3 & -\frac{2}{3} & -\frac{2}{3} & \frac{2}{3} 
            },\quad 
{\bf L}=\bordermatrix{     & v_1 & v_2 & v_3 \cr
               e_1 &  0  & -\sqrt{\frac{2}{3}}\,i & 0 \cr
               e_2 & \sqrt{\frac{2}{3}}\,i & 0 & 0 \cr
               e_3 & 0 & 0 & -\sqrt{\frac{2}{3}}\,j \cr
               e_4 & 0 & \sqrt{\frac{2}{3}}\,j & 0 \cr
               e_5 & -\sqrt{\frac{2}{3}}\,k & 0 & 0 \cr
               e_6 & 0 & 0 & \sqrt{\frac{2}{3}}\,k \cr
               e_7 & \sqrt{\frac{2}{3}} & 0 & 0 \cr
               e_8 & 0 & \sqrt{\frac{2}{3}} & 0 \cr
               e_9 & 0 & 0 & \sqrt{\frac{2}{3}} 
            },
\end{split}
\end{equation*}
whence $\Spec(\psi(\tilde{\bf W}))$ and $\Spec(\psi({\bf U}))$ turn out to be
\begin{equation*}
\begin{split}
&\Spec(\psi(\tilde{\bf W}))=\Bigg{\{} -\dfrac{2}{3},\,-\dfrac{2}{3},\,
\dfrac{4}{3},\,\dfrac{4}{3},\,\dfrac{4}{3},\,\dfrac{4}{3} \Bigg{\}},\\
&\Spec(\psi({\bf U}))=\Bigg{\{} -\dfrac{1}{3}{\pm}\dfrac{2\sqrt{2}}{3}i,\,
-\dfrac{1}{3}{\pm}\dfrac{2\sqrt{2}}{3}i,\,
\dfrac{2}{3}{\pm}\dfrac{\sqrt{5}}{3}i,\,\dfrac{2}{3}{\pm}\dfrac{\sqrt{5}}{3}i,\,
\dfrac{2}{3}{\pm}\dfrac{\sqrt{5}}{3}i,\,\\
&\qquad \dfrac{2}{3}{\pm}\dfrac{\sqrt{5}}{3}i,\,-1,\,-1,\,-1,\,-1,\,-1,\,-1 \Bigg{\}}.
\end{split}
\end{equation*}
Thus we obtain: 
\begin{equation*}
\begin{split}
\sigma_r({\bf U})&=\{-1\}{\cup}
\Bigg{(}-\dfrac{1}{3}+\dfrac{2\sqrt{2}}{3}i\Bigg{)}^{\!\!\HM^*}
{\cup}\Bigg{(}\dfrac{2}{3}+\dfrac{\sqrt{5}}{3}i\Bigg{)}^{\!\!\HM^*}.
\end{split}
\end{equation*}
The minimal polynomial of ${\bf U}$ is given by 
\[
p^{({\bf U})}(x)=p_1^{({\bf U})}(x)p_2^{({\bf U})}(x)p_3^{({\bf U})}(x), 
\]
where $p_1^{({\bf U})}(x)=x^2+\frac{2}{3}x+1$, 
$p_2^{({\bf U})}(x)=x^2-\frac{4}{3}x+1$ and $p_3^{({\bf U})}(x)=x+1$. 
Hence we have the decomposition of $\HM^9$ into root subspaces: 
\[
\HM^9=\mathscr{M}_1{\oplus}\mathscr{M}_2{\oplus}\mathscr{M}_3
\]
By using Theorem \ref{EigenCorres}, we obtain two 
right eigenvectors corresponding to the right eigenvalue $\mu=-\frac{2}{3}$ of $\tilde{W}$: 
\[
{\bf v}_1=\begin{pmatrix}1\\1\\1\end{pmatrix},\,
{\bf v}_2=\begin{pmatrix}j\\j\\j\end{pmatrix}={\bf v}_1j.
\]
Then, using Theorem \ref{ThmEigenvectors} we obtain right eigenvectors corresponding to the 
right eigenvalue $\lambda_{\pm}=(-1{\pm}2\sqrt{2}i)/3$ of ${\bf U}$ as follows: 
\[
{\bf u}_1^\pm=\begin{pmatrix}{\pm}\sqrt{2}+i\\{\mp}\sqrt{2}-i\\j{\pm}\sqrt{2}k\\-j{\mp}\sqrt{2}k\\
{\mp}\sqrt{2}j+k\\{\pm}\sqrt{2}j-k\\2{\pm}\sqrt{2}i\\2{\pm}\sqrt{2}i\\2{\pm}\sqrt{2}i\end{pmatrix},\,
{\bf u}_2^\pm=\begin{pmatrix}{\mp}\sqrt{2}j+k\\{\pm}\sqrt{2}j-k\\-1{\pm}\sqrt{2}i\\1{\mp}\sqrt{2}i\\
{\mp}\sqrt{2}-i\\{\pm}\sqrt{2}+i\\2j{\mp}\sqrt{2}k\\2j{\mp}\sqrt{2}k\\2j{\mp}\sqrt{2}k\end{pmatrix}
={\bf u}_1^{\mp}j,
\]
where ${\bf u}_r^+$ (resp. ${\bf u}_r^-$) is a right eigenvector corresponding to 
$\lambda_+$ (resp. $\lambda_-$) obtained from 
${\bf v}_r$ by using Theorem \ref{ThmEigenvectors} for $r=1,2$. 
${\bf u}_1^\pm,\,{\bf u}_2^\pm{\;\in\;}\mathscr{M}_1$, and  
one can check by using Proposition \ref{ProHLinearIndep} that ${\bf u}_1^+$ and ${\bf u}_2^+$ 
(resp. ${\bf u}_1^-$ and ${\bf u}_2^-$) are $\HM$-linearly independent. 

In the same fashion, we obtain four 
right eigenvectors corresponding to the right eigenvalue $\mu=\frac{4}{3}$ of $\tilde{W}$: 
\[
{\bf v}_3=\begin{pmatrix}1\\0\\-1\end{pmatrix},\,
{\bf v}_4=\begin{pmatrix}0\\1\\-1\end{pmatrix},\,
{\bf v}_5=\begin{pmatrix}j\\0\\-j\end{pmatrix}={\bf v}_3j,\,
{\bf v}_6=\begin{pmatrix}0\\j\\-j\end{pmatrix}={\bf v}_4j.
\]
Using Theorem \ref{ThmEigenvectors} again, we obtain right eigenvectors corresponding to the 
right eigenvalue $\lambda_{\pm}=(2{\pm}\sqrt{5}i)/3$ of ${\bf U}$ as follows: 
\begin{equation*}\begin{split}
&{\bf u}_3^\pm=\begin{pmatrix}3i\\{\mp}\sqrt{5}-2i\\-2j{\mp}\sqrt{5}k\\
3j\\{\mp}\sqrt{5}j-k\\{\mp}\sqrt{5}j-k\\1{\pm}\sqrt{5}i\\0\\-1{\mp}\sqrt{5}i\end{pmatrix},\qquad\quad\;\,
{\bf u}_4^\pm=\begin{pmatrix}{\pm}\sqrt{5}+2i\\-3i\\j{\mp}\sqrt{5}k\\j{\mp}\sqrt{5}k\\
-3k\\{\mp}\sqrt{5}j+2k\\0\\1{\pm}\sqrt{5}i\\-1{\mp}\sqrt{5}i\end{pmatrix},\; \\
&{\bf u}_5^\pm=\begin{pmatrix}3k\\{\pm}\sqrt{5}j-2k\\2{\mp}\sqrt{5}i\\-3\\
{\mp}\sqrt{5}+i\\{\mp}\sqrt{5}+i\\j{\mp}\sqrt{5}k\\0\\-j{\pm}\sqrt{5}k\end{pmatrix}
={\bf u}_3^{\mp}j,\;
{\bf u}_6^\pm=\begin{pmatrix}{\mp}\sqrt{5}j+2k\\-3k\\-1{\mp}\sqrt{5}i\\-1{\mp}\sqrt{5}i\\
3i\\{\mp}\sqrt{5}-2i\\0\\j{\mp}\sqrt{5}k\\-j{\pm}\sqrt{5}k\end{pmatrix}={\bf u}_4^{\mp}j,
\end{split}\end{equation*}
where ${\bf u}_r^+$ (resp. ${\bf u}_r^-$) is a right eigenvector corresponding to 
$\lambda_+$ (resp. $\lambda_-$) obtained from 
${\bf v}_r$ by using Theorem \ref{ThmEigenvectors} for $r=3,4,5,6$. 
${\bf u}_3^\pm,\,{\bf u}_4^\pm,\,{\bf u}_5^\pm,\,{\bf u}_6^\pm{\;\in\;}\mathscr{M}_2$, and 
it can be shown that ${\bf u}_3^+,\,{\bf u}_4^+,\,{\bf u}_5^+,\,{\bf u}_6^+$ 
(resp. ${\bf u}_3^-,\,{\bf u}_4^-,\,{\bf u}_5^-,\,{\bf u}_6^-$) are $\HM$-linearly independent. 

For the remaining right eigenvalue $-1$, we can find three $\HM$-linearly independent eigenvectors 
such as 
\[
{\bf u}_7=\begin{pmatrix}i\\i\\0\\0\\0\\0\\-1\\1\\0\end{pmatrix},\,
{\bf u}_8=\begin{pmatrix}j\\j\\-i\\-i\\1\\1\\0\\0\\0\end{pmatrix},\,
{\bf u}_9=\begin{pmatrix}i\\i\\j\\j\\0\\0\\-1\\0\\1\end{pmatrix}.
\]
It follows that
\[
\mathscr{M}_1=\bigoplus_{k=1}^2{\bf u}_k^+\HM,\;\mathscr{M}_2=\bigoplus_{k=3}^6{\bf u}_k^+\HM,\;
\mathscr{M}_3=\bigoplus_{k=7}^9{\bf u}_k\HM.
\]
\end{exm}

\par\noindent
{\bf Acknowledgment.} N. Konno is partially supported by the Grant-in-Aid for 
Scientific Research (Challenging Exploratory Research) of Japan Society for the Promotion 
of Science (Grant No. 15K13443). K. Matsue was partially supported by Program for Promoting 
the reform of national universities (Kyusyu University), Ministry of Education, Culture, Sports, Science and Technology (MEXT), Japan, and World Premier International Research Center Initiative (WPI), MEXT, Japan. 
H. Mitsuhashi and I. Sato are partially supported by the Grant-in-Aid 
for Scientific Research (C) of Japan Society for the Promotion of Science (Grant No. 16K05249 and 15K04985, respectively). 
We are grateful to Yu. Higuchi for fruitful discussions during the course of preparing this work. 
We thank K. Tamano, S. Matsutani, Y. Ide, M. Kosuda and O. Kada for helpful comments on early versions of this paper.

\begin{small}
\bibliographystyle{jplain}

\end{small}

\end{document}